\documentclass[preprint,12pt]{elsarticle}

\usepackage{amssymb}
\usepackage{bm}
\usepackage[ruled,vlined,linesnumbered,commentsnumbered]{algorithm2e}
\usepackage{enumitem}

\newcommand{\textem}[1] {{\em #1}}
\newcommand\msize[1]{\left|#1\right|}

\newcommand\mset[1]{\{#1\}}
\newcommand\mseq[1]{\langle#1\rangle}

\newcommand\order{\mathcal{O}}
\newenvironment{mycase}[2]{%
\vspace{3mm} \noindent \textbf{Case~#1}: #2\par}{}

\newenvironment{mycaselast}[2]{%
\vspace{3mm} \noindent \textbf{Case~#1}: #2\par}{\vspace{3mm}}

\newenvironment{listing}[1]{%
        \begin{list}{*}{%
                 \settowidth{\labelwidth}{#1}%
                 \setlength{\leftmargin}{\labelwidth}%
                  \advance \leftmargin by 12pt
                   \setlength{\itemsep}{0pt}%
                   \setlength{\parsep}{0pt}%
                   \setlength{\topsep}{0pt}%
                   \setlength{\parskip}{0pt}%
}%
}{%
\end{list}}

\newcommand\vertices{V}

\newcommand\inside{\mathsf{in}}
\newcommand\outside{\mathsf{out}}

\newcommand\parent{\mathsf{par}}

\newcommand\largest{\mathsf{larg}}
\newcommand\insertion{\mathsf{ins}}
\newcommand\lowerleft{\mathsf{left}}

\newcommand\setkout{\mathcal{S}_{\leq k}}

\usepackage{color}
\definecolor{darkred}{rgb}{0.8,0,0}

\newcommand\convexhull{\mathsf{CH}}

\newcommand\ftree{\mathcal{T}_{\leq k}}
\newcommand\prede{\mathsf{pred}}
\newcommand\succe{\mathsf{succ}}
\newcommand\embed{\mathsf{emb}}

\newcommand\dig{\mathsf{dig}}
\newcommand\rmv{\mathsf{rmv}}

\newcommand\pseq{\mathsf{PS}}

\newcommand\tang{t}
\newcommand\inspset{I}
\newcommand\region{R}
\newcommand\regnum{R_{\#}}
\newcommand\area{\mathsf{area}}

\usepackage{amsthm}
\usepackage{amsmath}
\newtheorem{theorem}{Theorem}[section]
\newtheorem{corollary}{Corollary}[theorem]
\newtheorem{lemma}[theorem]{Lemma}

\journal{arXiv}

\begin{document}
\begin{frontmatter}

\title{Efficient Enumeration of At Most $k$-Out Polygons}

\author[inst1]{Waseem Akram}

\affiliation[inst1]{organization={Indian Institute of Technology},
            country={India}}

\author[inst2]{Katsuhisa Yamanaka}
\affiliation[inst2]{organization={Iwate University},
            country={Japan}}

\begin{abstract}
Let $S$ be a set of $n$ points in the Euclidean plane and general position i.e., no three points are collinear.
An \emph{at most $k$-out polygon of $S$} is a simple polygon such that each vertex is a point in $S$ and there are at most $k$ points outside the polygon.
In this paper, we consider the problem of enumerating all the at most $k$-out polygon of $S$.
We propose a new enumeration algorithm for the at most $k$-out polygons of a point set.
Our algorithm enumerates all the at most $k$-out polygons in $\order(n^2 \log{n})$ delay, while the running time of an existing algorithm is $\order(n^3 \log{n})$ delay.
\end{abstract}


\begin{keyword}
enumeration algorithm \sep simple polygons \sep $k$-out polygons
\end{keyword}

\end{frontmatter}


\section{Introduction}
\label{sec:intro}

Enumeration problems (or listing) are fundamental and important in computer science and have applications in various domains, including bioinformatics and artificial intelligence.
A lot of enumeration algorithms for enumeration problems have been proposed~\cite{Wasa16}.
Among them, enumeration problems for geometric objects have been studied. 
For example, enumeration algorithms are known for triangulations~\cite{AvisF96,B02,KT09,Wettstein17}, non-crossing spanning trees~\cite{AvisF96,KT09,NakahataHMY20,Wettstein17},
non-crossing spanning cycles~\cite{NakahataHMY20,Wettstein17},
non-crossing convex partitions~\cite{Wettstein17},
non-crossing perfect matchings~\cite{Wettstein17},
non-crossing convex subdivisions~\cite{Wettstein17},
pseudoline arrangements~\cite{YamanakaNMUN09}, unfoldings of Platonic solids~\cite{HoriyamaS11}, Archimedian solids~\cite{HoriyamaMS18}, and so on.

A \emph{simple polygon of a point set} is a simple polygon such that every vertex is a point of the point set.
In this paper, we focus on  enumeration problems
of simple polygons of a point set.
For several classes of simple polygons of a point set, enumeration problems have been studied.
We are given a set $S$ of $n$ points in the Euclidean plane and general position.
A \emph{surrounding polygon} of $S$ is a simple polygon such that
every point in $S$ is either a vertex of the polygon or inside the polygon.
The class of surrounding polygons was proposed by Yamanaka\emph{~et~al.}~\cite{YamanakaAHOUY21} and they proposed an enumeration algorithm which enumerates in $\order(n^2 \log n)$ time for each.
The running time was improved to $\order(n^2)$ time for each~\cite{TeruiYHHKU23}.
A surrounding polygon includes an important class of simple polygons.
A surrounding polygon is called a \textem{non-crossing spanning cycle} of $S$ if every point in $S$ is a vertex of the polygon.
The non-crossing spanning cycles of a point set are known as appealing objects in the area of computational geometry
and have been studied in the contexts of counting~\cite{MarxM16,NakahataHMY20,Wettstein17}, random generation~\cite{AuerH96,Sohler99,TeramotoMUA06,ZhuSSM96}, and enumeration~\cite{NakahataHMY20,Wettstein17}.
For the enumeration problem of non-crossing spanning cycles, it was an open problem whether there exists an output-polynomial time algorithm.\footnote{
The running time of an enumeration algorithm is said to be \emph{output-polynomial} if it is polynomial in the input size and the output size of the algorithm.
}
Recently, Eppstein~\cite{Eppstein24} showed that non-crossing spanning cycles 
can be enumerated in output-polynomial time.
Terui\emph{~et~al.}~\cite{TeruiYHHKU23} proposed another class of simple polygons. 
A simple polygon of a point set is an \emph{empty polygon} if every point in $S$ is either a vertex of the polygon or outside the polygon.
They proposed an algorithm that enumerates all the empty polygons of a given point set in $\order(n^2)$ time for each.
The empty polygons are a generalization of empty convex polygons, which are empty polygons of a point set such that the polygons are convex.
The empty convex polygons have been studied in the contexts of counting~\cite{Bae22,MitchellRSW95,RoteWZW91,RoteW92} and enumeration~\cite{DobkinEO90}.

Recently, we introduced a new class of simple polygons, called at most $k$-out polygons, of a point set~\cite{WaseemY24}.
A simple polygon $P$ of a point set $S$ is an \emph{at most $k$-out polygon} if there are at most $k$ points outside $P$ and other points are either vertices of $P$ or inside $P$.
See \figurename~\ref{fig:k-out} for examples.
The class of at most $k$-out polygons is a generalization of the class of surrounding polygons in the sense that the set of at most $k$-out polygons of a point set $S$ coincides with (1) the set of surrounding polygons of $S$ when $k=0$ and (2) the set of simple polygons of $S$ when $k=n-3$.
\begin{figure}[tb]
	\centering
	\includegraphics[scale=.5]{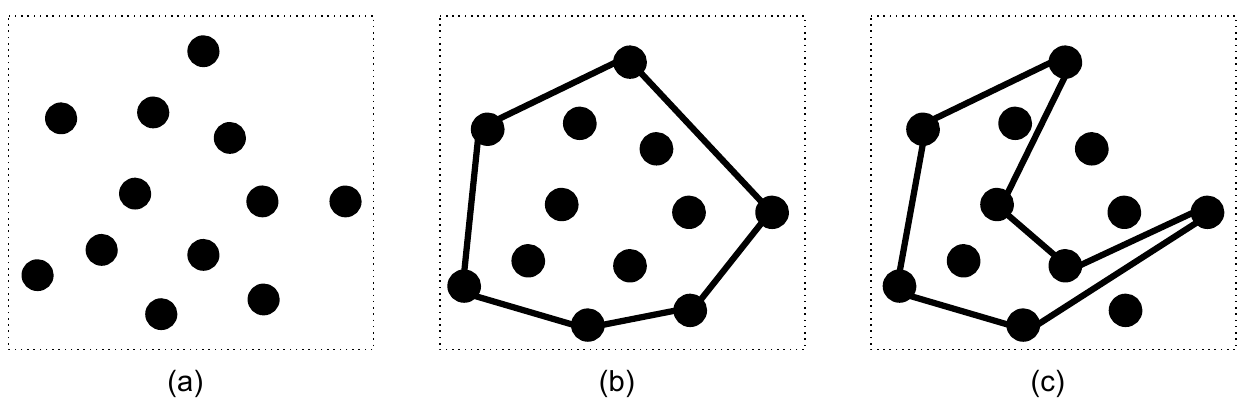}
	\caption{(a) A given point set $S$. (b) An at most $3$-out polygon of $S$ without any outside point~(the convex hull). (c) An at most $3$-out polygon of $S$ with 3 outside points.}
 \label{fig:k-out}
\end{figure}
We design an algorithm that enumerates all the at most $k$-out polygons in $\order(n^2 \log{n})$ delay
and $\order(n^2)$ space, 
while the running time of an existing enumeration algorithm~\cite{WaseemY24} is $\order(n^3 \log{n})$ delay. 
Our enumeration algorithm is based on the reverse-search technique by Avis and Fukuda~\cite{AvisF96}.
Since simple polygons are a particular instance of at most $k$-out polygons, our algorithm can be used to enumerate all simple polygons of the point set $S$. To the best of our knowledge, there is no non-trivial enumeration algorithm for simple polygons.

\section{Preliminaries}

In this section, we define some notations required in this paper.

Let $S$ be a set of $n$ points in the Euclidean plane.
Throughout this paper, we assume that $S$ is in general position, i.e., no three points are collinear.
For a point $p \in S$, $x(p)$ and $y(p)$ denote the $x$- and $y$-coordinate values of $p$, respectively.
For two points $p,q$ in $S$,
we denote $p <_{\mathrm{xy}} q$ if $x(p) < x(q)$ or $x(p)=x(q)$ and $y(p)<y(q)$.
A point $p$ is the \emph{lower-left} in $S'\subseteq S$ if $p <_{\mathrm{xy}} p'$ for every $p' \in S' \setminus \mset{p}$.
Note that the lower-left point of a point set is unique.
For three points $p,q,r \in S$, we denote the region (properly) enclosed by the triangle consisting of the points by $\triangle(p,q,r)$.

A \textem{simple polygon} is a closed region of the plane
enclosed by a simple cycle of line segments.
Here, a simple cycle means that
two adjacent line segments intersect only at their
common endpoint and no two non-adjacent line segments
intersect.
A sequence $P = \mseq{p_1,p_2,\ldots,p_t}$, $(t\leq n)$, of points in $S$ is a \emph{simple polygon of $S$} if the alternating sequence of points and line segments
$$
\mseq{p_1,(p_1,p_2),p_2,(p_2,p_3), \ldots , p_t,(p_t,p_1)}
$$
forms a simple polygon.
Let $P = \mseq{p_1,p_2,\ldots,p_t}$ be a simple polygon of $S$.
We denote the set of vertices of $P$ by $\vertices(P)$.
We suppose that the vertices on $P$ appear in counterclockwise order starting from the lower-left vertex $p_1$ of $\vertices(P)$.
We denote by $\inside(P) \subseteq S$ and $\outside(P) \subseteq S$ the sets of the points inside and outside $P$, respectively.
We denote by $p_i \prec p_j$ if $i < j$ holds,
and we say that $p_j$ is \emph{larger than} $p_i$ on $P$.
$\prede(p_i)$ and $\succe(p_i)$ denote the predecessor and successor of
$p_i$ of $P$, respectively.
Note that the successor of $p_t$ is $p_1$.
A vertex $p_i$ of $P$ is 
\emph{embeddable} if the triangle $\triangle(\prede(p_i),p_i,\succe(p_i))$ does not intersect the interior of $P$.
An \emph{embedment} of an embeddable point $p_i$ of $P$
is to remove two edges $(\prede(p_i),p_i)$ and $(p_i,\succe(p_i))$ and insert the edge $(\prede(p_i),\succe(p_i))$.
Note that we have a simple polygon after an embedment operation to an embeddable point. 
We denote by $\embed(P,p_i)$ the simple polygon
obtained from $P$ by applying the embedment of $p_i$ to $P$. Note that some points of $\outside(P)$ may be included in the interior of the resulting polygon $\embed(P,p_i)$.
See \figurename~\ref{fig:embedment} for examples.
\begin{figure}[tb]
	\centering
	\includegraphics[scale=.5]{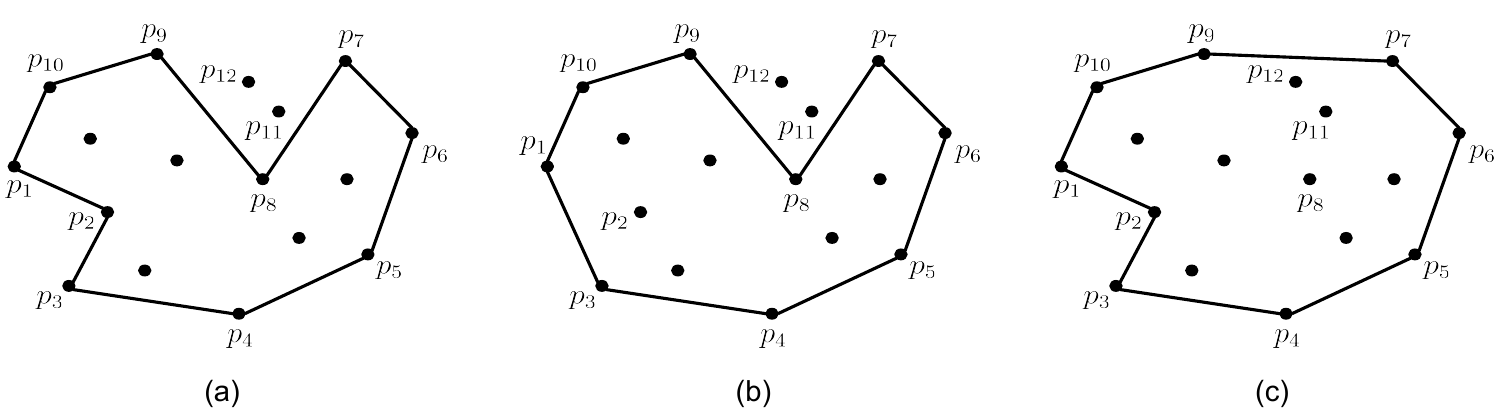}
	\caption{(a) A simple polygon of a point set, where $p_2$ and $p_8$ are embeddable points. (b) The simple polygon obtained from the polygon of (a) by applying an embedment to $p_2$. 
    (c) The simple polygon of $S$ obtained from the polygon of (a) by applying an embedment to $p_8$.}
    \label{fig:embedment}
\end{figure}

We say that a simple polygon of $S$ is \emph{convex} if every vertex of the polygon is non-embeddable.
The \emph{convex hull} of $S$, denoted by $\convexhull(S)$, is the simple polygon with the smallest area that contains all the points in $S$. 
Note that the convex hull of $S$ is a convex polygon of $S$.
Let $P$ be a convex polygon of $S$, and let $p_i$ be a vertex of $P$. 
A line $l$ passing through $p_i$ \emph{supports} $P$ if the entire $P$ lies on one side of the line $l$ \cite{Preparata12}. 
The line $l$ is called a \emph{tangent} to $P$ at $p_i$. 
See Figure~\ref{fig:tangents} for an illustration. 
For a point $p\in \outside(P)$, a vertex $p_i$ of $P$ is a \emph{tangency vertex} for $p$ if the line containing the line segment $(p,p_i)$ supports $P$.
The line is called the \emph{tangent line} from $p$ to $P$.
Note that there are exactly two tangent lines (and two tangency vertices) from any point in $\outside(P)$.

\begin{lemma}[\cite{Preparata79}]
Let $P$ be a convex polygon of $S$.
One can preprocess the polygon $P$ in $\order(n\log n)$ time and $\order(n)$ space so that, given a point $p\in \outside(P)$, the tangent lines from $p$ to $P$ can be computed in $\order(\log n)$ time.
\label{lem:tangents}
\end{lemma}
\begin{figure}[tb]
    \centering\includegraphics[scale=.55]{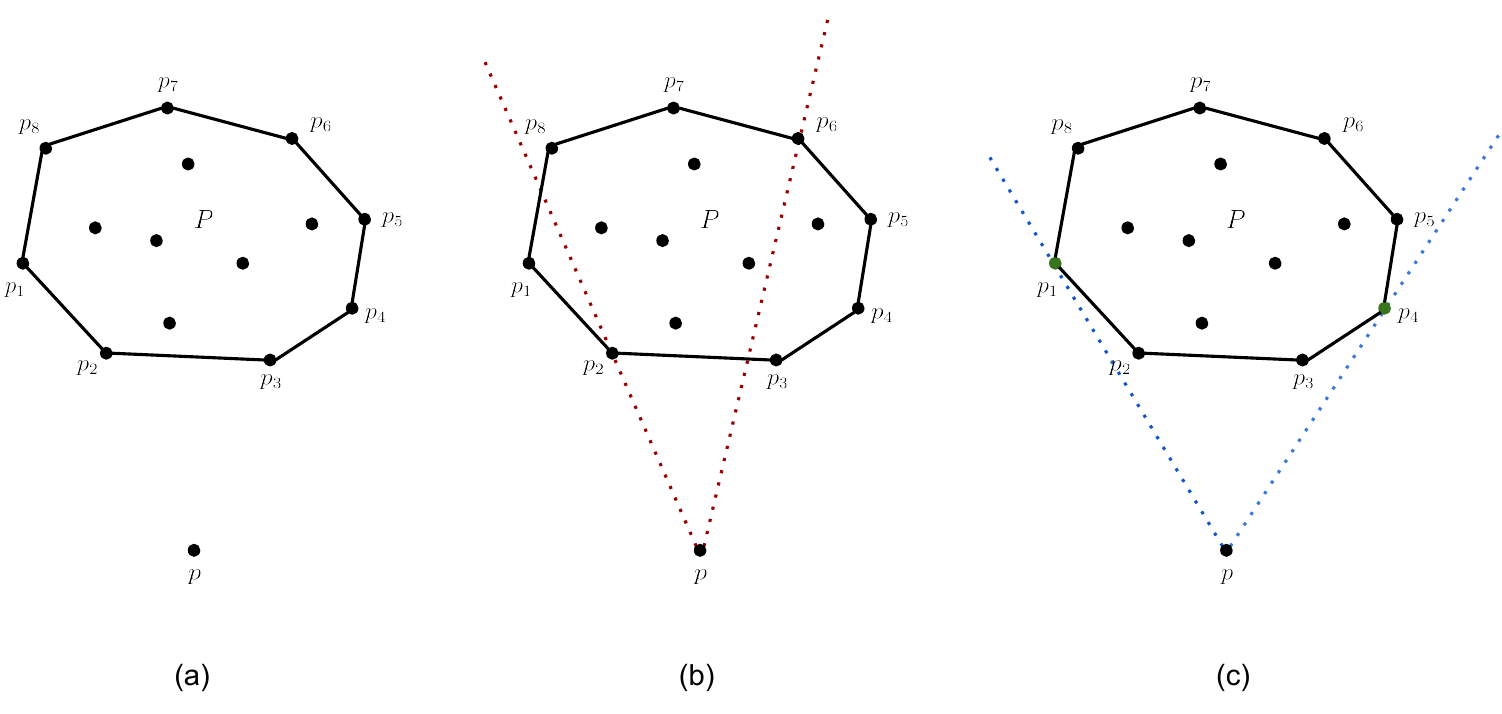}
    \caption{(a) A simple polygon $P$ with an outside point $p$. (b) The line through $p$ and $p_2$ (resp. $p_6$) is not a tangent line at $p_2$ (resp. $p_6$). (c) The lines defined by $(p,p_1)$ and $(p,p_4)$ are two tangent lines of $P$ passing through $p$.}
    \label{fig:tangents}
\end{figure}
%
%
We denote the two tangency vertices of $P$ for $p \in \outside(P)$ by $\tang_1(P,p)$ and $\tang_2(P,p)$.
We denote the region bounded by the two line segments $(p, \tang_1(P,p))$, $(p,\tang_2(P,p))$ and the edges of $P$ inside $\triangle(p,\tang_1(P,p),\tang_2(P,p))$ or on the boundary of the triangle by $\region(P,p)$.
Note that, if $\tang_1(P,p)$ is adjacent to $\tang_2(P,p)$ by an edge of $P$, from the definition, we have $\region(P,p)=\triangle(p,\tang_1(P,p),\tang_2(P,p))$.
The point $p \in \outside(P)$ is \emph{insertable} to $P$ if $\region(P,p)$ does not include any point in $\outside(P)$.
See \figurename~\ref{fig:insertable}.
An \emph{insertion} of an insertable point $p$ to $P$ is to remove the edges inside $\triangle(p,\tang_1(P,p),\tang_2(P,p))$ or on the boundary of the triangle
and insert the segments $(p, \tang_1(P,p))$ and $(p,\tang_2(P,p))$ as two new edges. 
See \figurename~\ref{fig:pc-ins} for an example.
We denote by $\insertion(P,p)$ the simple polygon obtained from $P$ by applying the insertion of $p$.
\begin{figure}[h]
    \centering
    \includegraphics[scale=.65]{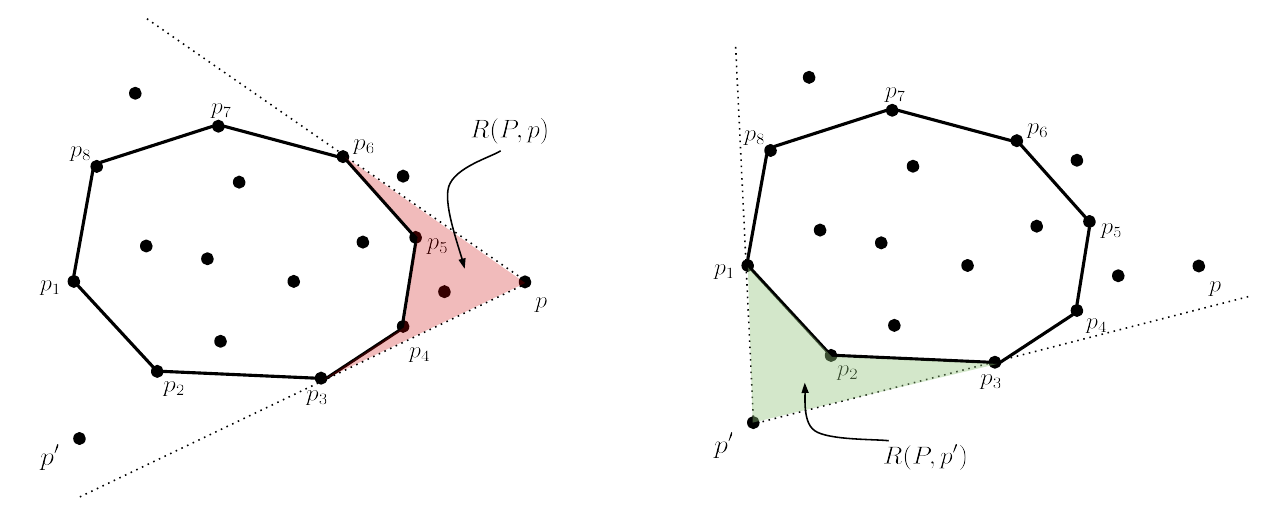}
    \caption{(a) The point $p$ is not insertable to the polygon $P$ as the region $R(P,p)$ (colored red) contains an outside point. (b) The point $p'$ is insertable to $P$.}
    \label{fig:insertable}
\end{figure}

\begin{figure}[h]
\centering
\includegraphics[scale=.35]{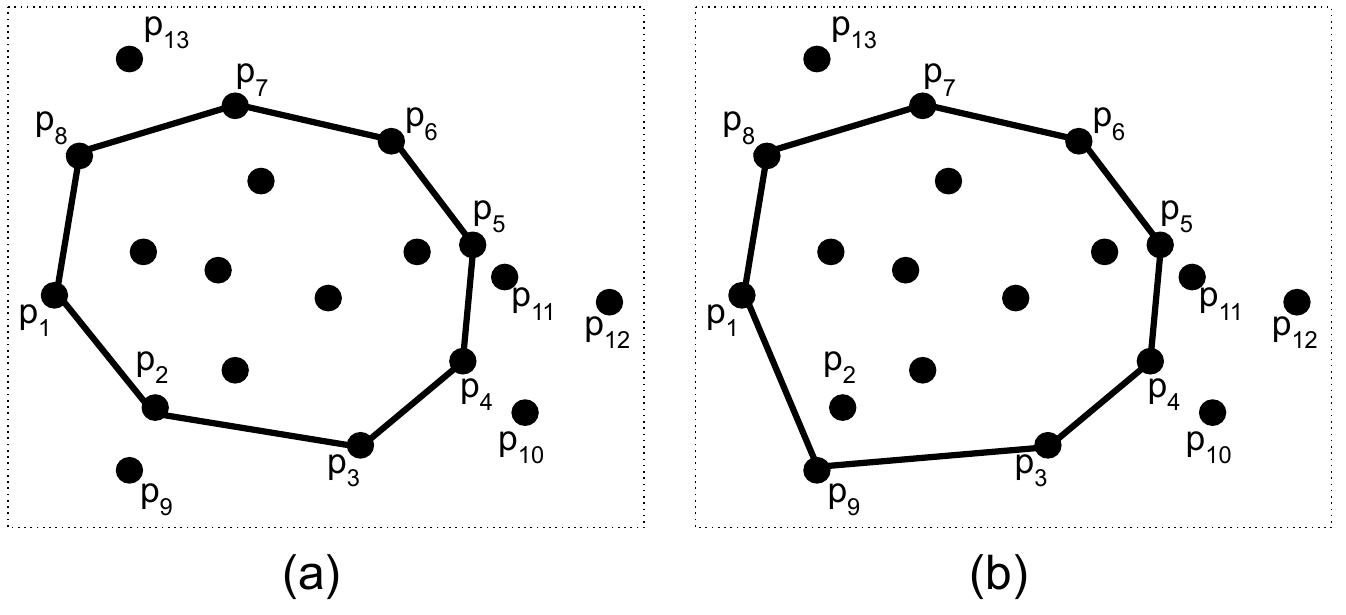}
\caption{(a) A convex polygon of a point set with $\outside(P)=\mset{p_9, p_{10}, p_{11}, p_{12}, p_{13}}$. (b) The polygon obtained from the polygon of (a) by inserting $p_{9}$.}
\label{fig:pc-ins}
\end{figure}

\begin{lemma} \label{lem:existence_ins}
Let $S$ be a set of points in the Euclidean plane,
and let $P$ be a convex polygon of $S$ except 
%
%
$\convexhull(S)$.
Then, there exists an insertable point of $P$.
\end{lemma}
\begin{proof}
We can find an insertable point of $P$ in the following way.
Let $p$ be a point in $\outside(P)$.
If $p$ is insertable, then we are done.
Now, we assume otherwise.
Then, there is a point inside $\region(P,p)$.
Let $p'$ be a point inside $\region(P,p)$.
If $p'$ is insertable, then we are done.
Otherwise, we do the same process in $\region(P,p')$.
Here, note that $\region(P,p')$ is included in $\region(P,p)$.
By repeating the above process, we finally have an insertable point of $P$.
\end{proof}

We denote by $\inspset(P)$ the set of insertable points of $P$.
We denote the lower-left point of $\inspset(P)$ as $\lowerleft(\inspset(P))$.

A simple polygon $P$ of $S$ is an \emph{at most $k$-out polygon} of $S$ if $\msize{\outside(P)} \leq k$ holds.
\figurename~\ref{fig:k-out} shows examples of at most $k$-out polygons.
We denote the set of the at most $k$-out polygons of $S$ by $\setkout(S)$.

%
%
\section{Family tree of at most $k$-out polygons}\label{sec:ft}

Let $S$ be a set of $n$ points in the Euclidean plane. 
In this section, we define a rooted tree structure on $\setkout(S)$, called a family tree.
By traversing the family tree, we enumerate all the polygons in $\setkout(S)$.
In Section~\ref{sec:algo}, we design an algorithm to traverse the family tree.

Let $P=\mseq{p_1,p_2,\ldots ,p_t}$ ($t\leq n$) be a polygon in $\setkout(S) \setminus \mset{\convexhull(S)}$.
Remember that $p_1$ is the lower-left vertex of 
$\vertices(P)$ and the vertices on $P$ are arranged in the counterclockwise order.
We denote by $\largest(P)$ the largest embeddable point of $P$. For convenience, we define $\largest(P)=\emptyset$ if $P$ is convex.
Now, to define a rooted tree structure,
we define the \emph{parent} $\parent(P)$ of $P$, as follows.

\begin{mycase}{1}{$P$ is not a convex polygon.}
In this case, $P$ has an embeddable point.
We define the \emph{parent} of $P$ as the polygon obtained by applying an embedment to the largest embeddable point $\largest(P)$. 
In \figurename~\ref{fig:embedment}, the polygon in \figurename~\ref{fig:embedment}(c) is the parent of the polygon in \figurename~\ref{fig:embedment}(a). 
However, the polygon in \figurename~\ref{fig:embedment}(b) is not the parent of the polygon \figurename~\ref{fig:embedment}(a).
\end{mycase}

\begin{mycaselast}{2}{$P$ is a convex polygon.}
Since $P \neq \convexhull(S)$, $\inspset(P) \neq \emptyset$ holds from Lemma~\ref{lem:existence_ins}.
Then, we define the \emph{parent} $\parent(P)$ of $P$ as the polygon obtained from $P$ by applying an insertion to the point $\lowerleft(\inspset(P))$ to $P$.
Remember that $\lowerleft(\inspset(P))$ is the lower-left point among the set of the insertable points to $P$.
See \figurename~\ref{fig:pc-ins} for an example.
In the figure, the polygon in \figurename~\ref{fig:pc-ins}(b) is the parent of the polygon in \figurename~\ref{fig:pc-ins}(a).
\end{mycaselast}

From the above case analysis, the parent $\parent(P)$ of $P$ is defined as follows. 
$$
\parent(P) = 
\begin{cases}
   \embed(P,\largest(P)) & \text{$P$ is not a convex polygon~(Case~1)} \\
   \insertion(P,\lowerleft(\inspset(P))) & \text{otherwise~(Case~2)}.
\end{cases}
$$
We say that $P$ is a \emph{child} of $\parent(P)$.
In Case~1, $\parent(P)$ is either a convex or non-convex polygon and we say that $P$ is a \emph{Type-1 child} of $\parent(P)$.
In Case~2, $\parent(P)$ is always a convex polygon and we say that $P$ is a \emph{Type-2 child} of $\parent(P)$.

\begin{lemma}\label{lem:par-unq}
Let $P$ be a polygon in $\setkout(S) \setminus \mset{\convexhull(S)}$, where $S$ is a set of points in Euclidean plane.
Then, the parent $\parent(P)$ of $P$ is an at most $k$-out polygon of $S$, always exists and is unique.
\end{lemma}
\begin{proof}
Immediate from the definition of the parents.
\end{proof}

The \emph{parent sequence} of a given polygon $P\in \setkout{(S)}$, denoted by $\pseq(P)$, is defined as the sequence of polygons obtained by recursively tracing the parents of $P$. 
More precisely,
$\pseq(P)=\langle P_1, P_2,\ldots, P_\ell\rangle$, where $P_1=P$ and $P_i=\parent(P_{i-1})$ for all $i=2,3,\ldots, \ell$. Figure~\ref{fig:pseq} illustrates an example.
\begin{figure}[h]
 	\centering
 	\includegraphics[scale=.5]{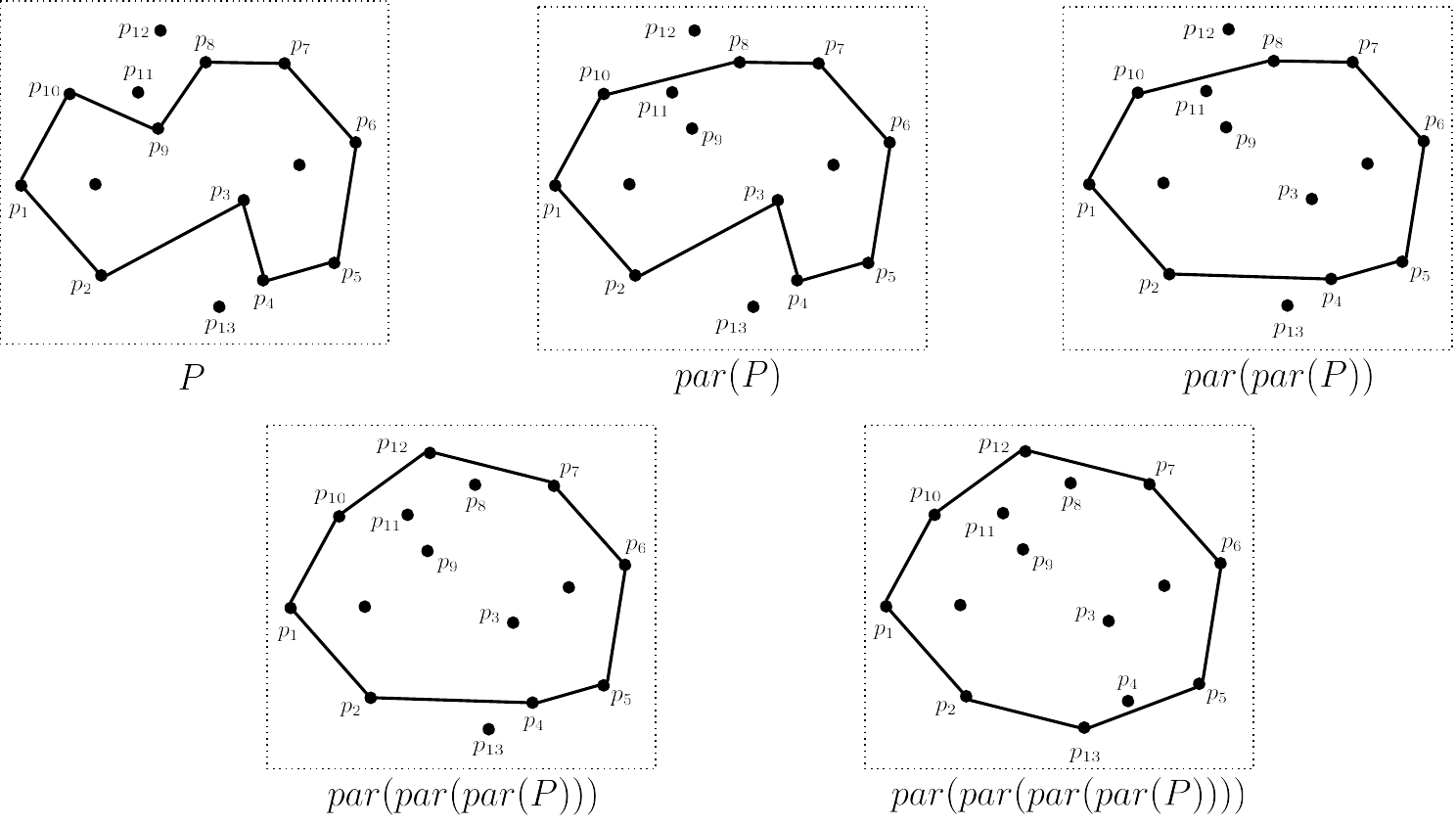}
 	\caption{The parent sequence $\pseq(P)$ of a polygon $P$.}
     \label{fig:pseq}   
 \end{figure}
\begin{lemma}\label{lem:pseq}
    For any polygon $P\in \setkout(S)$, the last polygon in the parent sequence $\pseq(P)$ is $\convexhull(S)$.
\end{lemma}
\begin{proof}
We denote the area of $P$ by $\area(P)$.
Then, we have $\area(P) < \area(\convexhull(S))$ for any polygon $P \in \setkout(S) \setminus \mset{\convexhull(S)}$.
From the definition of the parents, $\area(P) < \area(\parent(P))$ holds.
Hence, the claim holds.
\end{proof}

Lemma~\ref{lem:pseq} implies that we finally have the convex hull by finding parents repeatedly.
Hence, it can be observed that we have a tree structure rooted at $\convexhull(S)$ by merging the parent sequences of all the polygons of $\setkout(S)$.
We call such a tree the \emph{family tree} of $\setkout(S)$ and denote it by $\ftree(S)$. 
See Figure~\ref{fig:familytree} for an example.
\begin{figure}[h]
	\centering
	\includegraphics[width=14cm]{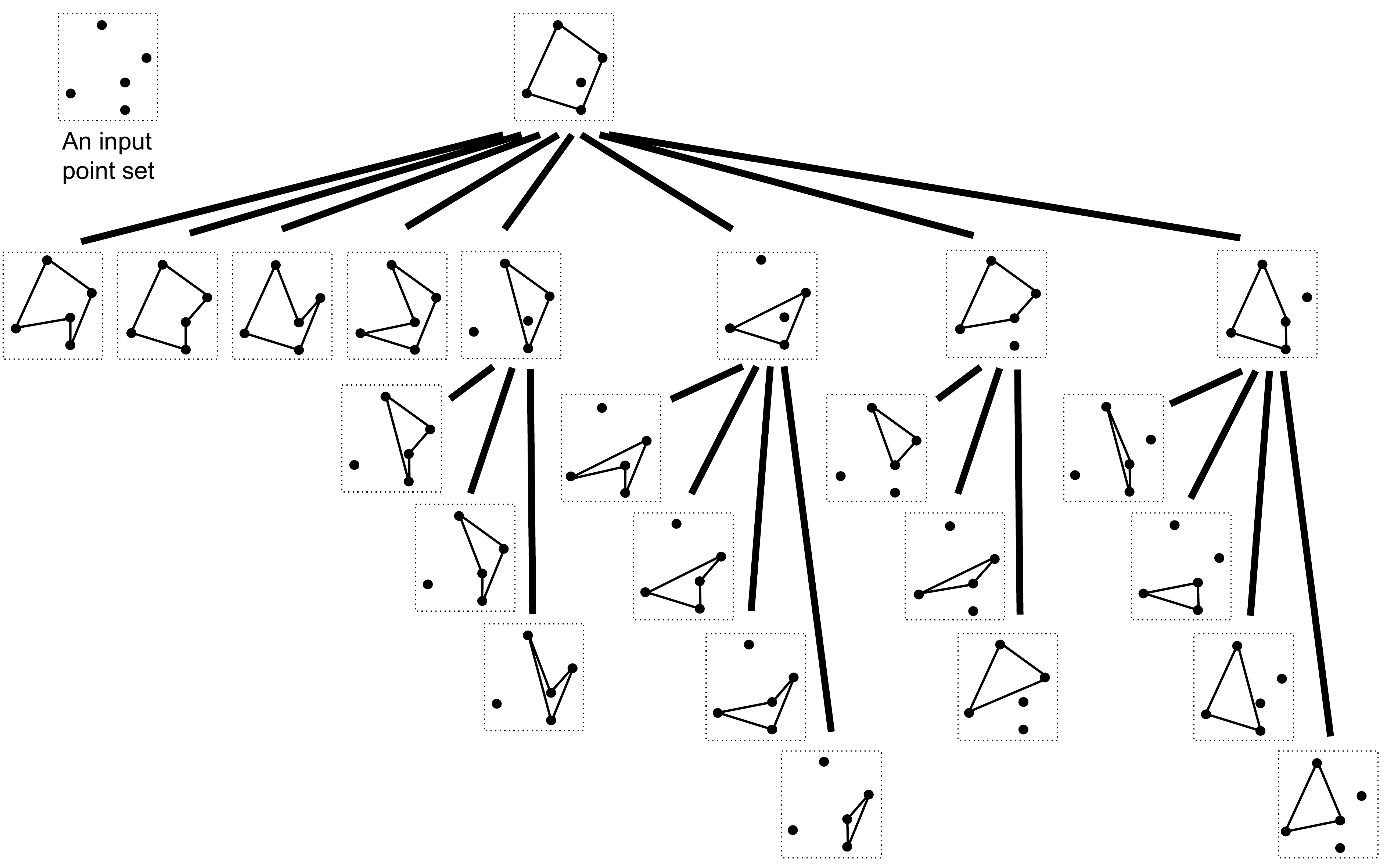}
	\caption{An example of a family tree of $\setkout(S)$,
		where the input point set $S$ is illustrated in the upper left of the figure and $k=2$.}
	\label{fig:familytree}
\end{figure}

\section{Enumeration algorithm of at most $k$-out polygons}\label{sec:algo}

In this section, we design a child-enumeration algorithm of a given at most $k$-out polygon.
By applying the child-enumeration algorithm recursively starting from the convex hull of a given point set, we traverse the family tree and enumerate all the at most $k$-out polygons.
%
Below, we show how to enumerate Type-1 and Type-2 children of $P$, respectively.

Let $S$ be a set of $n$ points in Euclidean plane, 
and let $P=\mseq{p_1,p_2,\ldots ,p_t}$ be a polygon in $\setkout(S)$.

\subsection{Type-1 children}
A pair $(p_i,p)$ of a vertex $p_i \in V(P)$ and a point $p\in \inside(P)$ is \emph{digable} if the triangle 
$\triangle(p,p_i,\succe(p_i))$
lies inside $P$ and there are at most $k-\msize{\outside(P)}$ points inside the triangle.
A \emph{dig operation} to a digable pair $(p_i,p)$ removes the edge $(p_i, \succe(p_i))$ and inserts the two edges $(p_i, p)$ and $(p, \succe(p_i))$.
We denote the resulting polygon by $\dig(P,p_i,p)$.
Note that $\dig(P,p_i,p)$ for a digable pair $(p_i,p)$ is also a polygon in $\setkout(S)$.
Intuitively, a dig operation makes an at most $k$-out polygon include an inside point as a vertex of $P$. 
The dig operation is the inverse of an embedment operation. 
The polygon $\dig(P,p_i,p)$ is possibly a child of $P$.
More precisely, $\dig(P,p_i,p)$ is a child of $P$ if $P=\parent(\dig(P,p_i,p))$ holds. 
We say that a digable pair $(p_i,p)$ is \emph{active} if $\dig(P,p_i,p)$ is a child of $P$.

If we have the set of the active pairs of $P$, 
it is easy to enumerate all the Type-1 children of $P$.
In this subsection, we show that one can construct such set in $\order(n^2 \log{n})$ time.
The following lemma rephrases the definition of the active pairs and is useful for checking the activity of a digable pair.

\begin{lemma}\label{lem:ischild}
	Let $P$ be a polygon in $\setkout(S)$, where $S$ is the set of $n$ points in Euclidean plane.
	Let $(p_i, p)$ be a digable pair, where $p_i$ is a vertex of $P$ and $p \in \inside(P)$. Then, $(p_i,p)$ is active 
    if and only if $p = \largest(\dig(P,p_i,p))$ holds.
	\label{lem:dig-act}
\end{lemma}

To efficiently enumerate children, 
we use the following two queries.
A \emph{triangular range query} of 3 points $p,q,r \in S$ asks the number of points in $\triangle(p,q,r)$.
The triangular range query can be answered efficiently.
\begin{lemma}[\cite{GoswamiDN2004}]\label{lem:triangle-query}
	A triangular query of a set of $n$ points can be answered in $\order(\log n)$ time with $\order(n^2)$-time preprocessing and $\order(n^2)$-additional space.
\end{lemma}

A \emph{ray shooting query} asks the first obstacle intersected by a query ray.
The ray shooting query can be answered efficiently.
\begin{lemma}[\cite{ChazelleEGGHSS94,ChazelleG89}]
	Let $S$ be a set of $n$ points.
	A ray shooting query of a polygon of $S$ can be answered in $\order(\log n)$ time with $\order(n\log{n})$-time preprocessing and $\order(n)$-additional space.
\end{lemma}

By using the triangular and ray shooting queries, we can check the activity of a pair $(p_i,p)$, as stated in the following lemma.

\begin{lemma}
	Let $P$ be an at most $k$-out polygon of $S$, where $S$ is the set of $n$ points. 
	Suppose that each of the triangular range and ray shooting queries can be computed in $\order(\log{n})$ time.
	Then, given $\largest(P)$ and
	a pair $(p_i, p)$, where $p_i \in V(P)$ and $p\in \inside(P)$,
	one can check whether $(p_i,p)$ is active in $\order(\log n)$ time.
	\label{lem:active_check_Type1}
\end{lemma}
\begin{proof}
	First, we check whether or not $(p_i,p)$ is digable.
	By ray shooting queries, we check that $\triangle(p_i,\succe(p_i),p)$ does not intersect with any edge of $P$.
	Then, by using a triangular query, we calculate the number of points in $\triangle(p_i,\succe(p_i),p)$ and check the number is less than or equal to $k-\msize{\outside(P)}$.
	These can be done in $\order(\log{n})$ time.
	
	If $(p_i,p)$ is digable, we do the following process to check whether the pair is active. 
	From now on, we assume that $(p_i,p)$ is digable.
	We have the following cases.
	
	\begin{mycase}{1}{$p_i \prec \prede(\largest(P))$.}
		If we apply a dig operation to $p_i$ smaller than $\prede(\largest(P))$, in the polygon $\dig(P,p_i,p)$, 
		the largest embeddable point is still $\largest(P)$, that is, $\largest(P) = \largest(\dig(P,p_i,p))$ holds.
		See \figurename~\ref{fig:active_cases}(b) for an example.
		Hence, from Lemma~\ref{lem:ischild}, $(p_i,p)$ is non-active and $\dig(P,p_i,p)$ is not a child of $P$.
	\end{mycase}
	
	\begin{mycase}{2}{$p_i = \prede(\largest(P))$.}
		If $\largest(P)$ is still embeddable in $\dig(P,p_i,p)$,
		then $(p_i,p)$ is non-active and $\dig(P,p_i,p)$ is not a child of $P$.
		See \figurename~\ref{fig:active_cases}(c).
		However, if $\largest(P)$ is non-embeddable in $\dig(P,p_i,p)$, then $(p_i,p)$ is active and $\dig(P,p_i,p)$ is a child.
		See \figurename~\ref{fig:active_cases}(d).
		Note that the embedability of $\largest(P)$ in $\dig(P,p_i,p)$ can be checked in $\order(\log{n})$ time using a ray shooting query.
	\end{mycase}
	
	\begin{mycaselast}{3}{$\prede(\largest(P)) \prec p_i$}
		In this case, $p=\largest(\dig(P,p_i,p))$ always holds.
		Thus, from Lemma~\ref{lem:ischild}, $(p_i,p)$ is active and $\dig(P,p_i,p)$ is a child of $P$.
		See \figurename~\ref{fig:active_cases}(e).
	\end{mycaselast}
	
	From the case analysis above, we have the following process to check whether or not $(p_i,p)$ is active.
	If $p_i \prec \prede(\largest(P))$, then $(p_i,p)$ is non-active.
	Let us assume that $p_i = \prede(\largest(P))$ holds.
	Then, we check $\largest(P)$ is still embeddable in $\dig(P,p_i,p)$ by using ray shooting query in $\order(\log{n})$ time.
	If $\largest(P)$ is still embeddable, then $(p_i,p)$ is non-active. Otherwise, it is active.
	If $\prede(\largest(P)) \prec p_i$, $(p_i,p)$ is always active.
	\begin{figure}[t]
		\centerline{\includegraphics[width=0.8\linewidth]{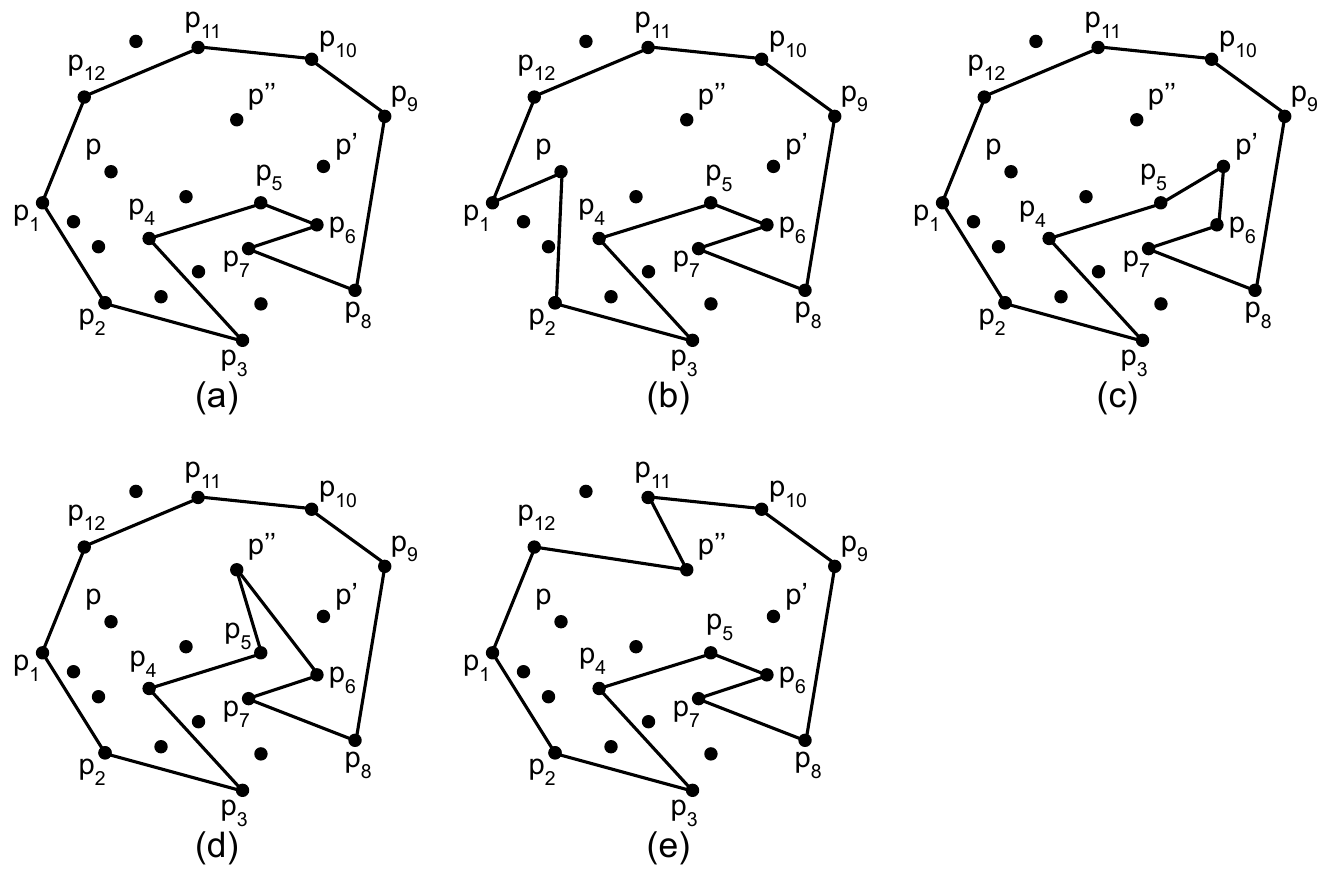}}
		\caption{(a) An at most 5-out polygon $P$, where $\largest(P) = p_6$.
			(b) $\dig(P,p_{1},p)$ is not a child of $P$, since $p$ is not the largest embeddable vertex.
			(c) $\dig(P,p_5,p')$ is not a child of $P$, since $p'$ is not the largest embeddable vertex.
			(d) $\dig(P,p_5,p'')$ is a child of $P$, since $p_6$ turns to be non-embeddable.
			(e) $\dig(P,p_{11},p'')$ is a child.}
		\label{fig:active_cases}
	\end{figure}
\end{proof}

From the lemma above, we have the following lemma.

\begin{corollary}\label{cor:active_set_Type1}
	Let $P$ be a polygon of $S \in \setkout(S)$, where $S$ is the set of $n$ points.
	Suppose that each of the triangular range and ray shooting queries can be computed in $\order(\log{n})$ time.
	One can construct the set of active pairs of $P$ in $\order(n^2 \log{n})$ time.
\end{corollary}
\begin{proof}
	From Lemma~\ref{lem:active_check_Type1},
	for a pair $(p_i,p)$, where $p_i \in V(P)$ and $p \in \inside(P)$, one can check whether or not it is active.
	Since there exist $\order(n^2)$ such pairs,
	we take $\order(n^2 \log{n})$ time.
\end{proof}

Finally, we estimate the running time to construct a Type-1 child from $P$ when an active pair is given.

\begin{lemma}\label{lem:child_const_Type1}
Let $P$ be a polygon in $\setkout(S)$, where $S$ is the set of $n$ points in Euclidean plane.
Then, given an active pair $(p_i,p)$, where $p_i \in V(P)$ and $p \in \outside(P)$, one can construct $\dig(P,p_i,p)$ in $\order(1)$ time.
\end{lemma}
\begin{proof}
	Constructing $\dig(P,p_i,p)$ from $P$ can be done in $\order(1)$ time.
\end{proof}

\subsection{Type-2 children}

Next, we consider how to enumerate Type-2 children.
Remember that only a convex polygon has Type-2 children.
Suppose that $P=\mseq{p_1,p_2,\ldots ,p_t}$ is a convex polygon in $\setkout(S)$.

%
%
Let $p_i$ be a vertex of $P$.
Now, we define a remove operation, as follows.
A \emph{remove operation} to a vertex $p_i$ of $P$ 
removes the two edges $(\prede(p_i), p_i)$  and $(p_i, \succe(p_i))$, and inserts the chain of the edges between $\prede(p_i)$ and $\succe(p_i)$ such that
the resulting polygon is the convex hull of $\inside(P) \cup V(P) \setminus \mset{p_i}$.
$\rmv(P,p_i)$ denotes the resulting polygon.
Notice that the remove operation is the inverse of the insert operation.
$\rmv(P,p_i)$ can be constructed in $\order(n\log n)$-time by a standard algorithm for finding convex hulls.
Note that $\rmv(P,p_i)$ is also a convex polygon in $\setkout(S)$ if $\msize{\outside(\rmv(P,p_i))}\leq k$. It can be observed that
$\rmv(P,p_i)$ is a child of $P$ if $P=\parent(\rmv(P,p_i))$. A vertex $p_i$ is called \emph{active} if $\rmv(P,p_i)$ is a child of $P$.

If we have the set of the active vertices of $P$, 
it is easy to enumerate all the Type-2 children of $P$.
In this subsection, we show that one can construct such set in $\order(n^2\log{n})$ time.
The following lemma is useful for checking the activity of a vertex.

\begin{lemma}\label{lem:rmv-act}
Let $P = \mseq{p_1,p_2,\ldots,p_t}$, $(t\leq n),$ be a convex polygon in $\setkout(S)$ with $\msize{\outside(P)}<k$, where $S$ is a set of $n$ points in Euclidean plane. 
Then, a vertex $p_i$ is active if and only if $p_i = \lowerleft(\inspset(\rmv(P,p_i)))$.
\end{lemma}

From Lemma~\ref{lem:rmv-act}, 
if $p_i = \lowerleft(\inspset(\rmv(P,p_i)))$, it implies that $p_i$ is active.
To check whether $p_i = \lowerleft(\inspset(\rmv(P,p_i)))$ efficiently, 
we maintain the the number of the points inside $\region(P,p)$ for each $p \in \outside(P)$.

%
%
For a polygon $P$ and point $p \in \outside(P)$, 
we denote the number of the points inside $\region(P,p)$ by $\regnum(P,p)$.
For three points $p,q,r \in S$, we denote the number of the points inside the triangle $\triangle(p,q,r)$ by $\triangle_{\#}(p,q,r)$.
Let $\rmv(P,p_i)$ be a Type-2 child of $P$.
Note that, for a point $p \in \outside(P)$, $p$ satisfies $\regnum(P,p) = 0$ if and only if $p \in \inspset(P)$.
To maintain $\regnum(P,p)$ for each $p \in \outside(P)$,
we investigate the difference between $\regnum(P,p)$ and $\regnum(\rmv(P,p_i),p)$.

First, we consider the case where $p = p_i$ holds.
Then, it is easy to observe that $\regnum(\rmv(P,p_i),p) = 0$ holds from the definition of the remove operation.
Next, we consider the case where $p_i=\tang_1(P,p)$ holds.
Remember that $\tang_1(P,p)$ and $\tang_2(P,p)$ are the tangency vertices of $p$ in $P$.
Then, $p$ has a new tangency vertex in $\rmv(P,p_i)$. (Note that $\tang_2(P,p)$ is still a tangency vertex of $p$ in $\rmv(P,p_i)$.)
We denote the new tangency vertex by $\tang_1'$.
Then, we can observe that 
$\regnum(\rmv(P,p_i),p)=\regnum(P,p) - \triangle_{\#}(p,\tang_1(P,p),\tang_1')$ holds.
See \figurename~\ref{fig:update_t1}.
\begin{figure}[tb]
\centerline{\includegraphics[width=0.8\linewidth]{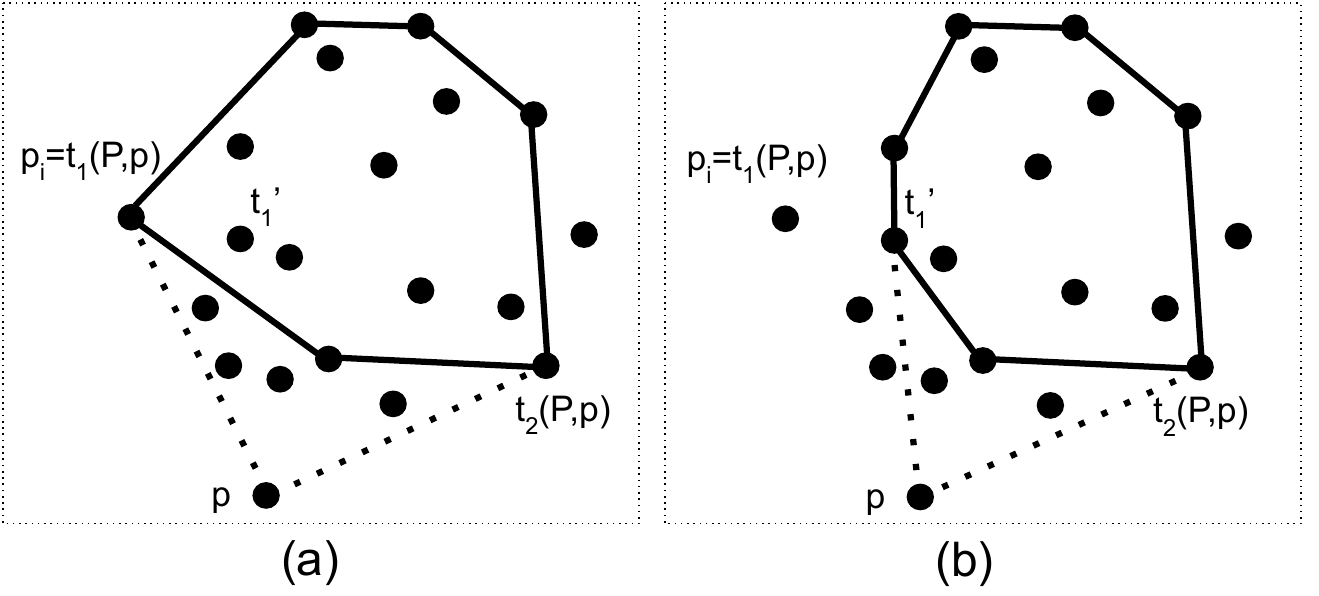}}
\caption{Illustration for the case $p_i=\tang_1(P,p)$. 
(a) $\regnum(P,p)=4$ in the polygon $P$.
(b) $\regnum(\rmv(P,p_i)) = \regnum(P,p) - \triangle_{\#}(p,\tang_1(P,p),\tang_1') =2$ in $\rmv(P,p_i)$.}
\label{fig:update_t1}
\end{figure}
Similarly, if $p_i=\tang_2(P,p)$, $p$ has a new tangency vertex in $\rmv(P,p_i)$.
We denote the new tangency vertex by $\tang_2'$.
It can be observed that $\regnum(\rmv(P,p_i),p)=\regnum(P,p) - \triangle_{\#}(p,\tang_2(P,p),\tang_2')$ holds.
Finally, we consider the case where $p \neq p_i$, $p_i\neq \tang_1(P,p)$, and $p_i\neq \tang_2(P,p)$.
Note that $\tang_1(P,p)$ and $\tang_2(P,p)$ are still tangency vertices of $p$ in $\rmv(P,p_i)$.
Then, if $p_i \in \triangle(p,\tang_1(P,p),\tang_2(P,p))$,
$\regnum(\rmv(P,p_i),p)=\regnum(P,p) +1$.
Otherwise, $\regnum(\rmv(P,p_i),p)=\regnum(P,p)$.
Therefore, we have the following equation, which is a relationship between $\regnum(P,p)$ and $\regnum(\rmv(P,p_i),p)$.

\begin{eqnarray}\label{eq:points}
& & \regnum(\rmv(P,p_i),p) \nonumber \\
& & = \begin{cases}
		0 & \text{$p=p_i$}, \\
		\regnum(P,p) - \triangle_{\#}(p,\tang_1(P,p),\tang_1') & \text{$p \neq p_i$ and $p_i = \tang_1(P,p)$}, \\
		\regnum(P,p) -\triangle_{\#}(p,\tang_2(P,p),\tang_2')) & \text{$p\neq p_i$ and $p_i = \tang_2(P,p)$},\\
		\regnum(P,p) +1 & \text{$p\neq p_i$, $p_i \neq \tang_1(P,p)$, $p_i \neq \tang_2(P,p)$,} \\
		& \text{and $p_i \in \triangle(p,\tang_1(P,p),\tang_2(P,p))$, and}\\
		\regnum(P,p) & \text{$p\neq p_i$, $p_i \neq \tang_1(P,p)$, $p_i \neq \tang_2(P,p)$,} \\
		& \text{and $p_i \notin \triangle(p,\tang_1(P,p),\tang_2(P,p))$}.
	\end{cases} \nonumber \\
\end{eqnarray}
\begin{lemma}\label{lem:update}
	Let $P=\mseq{p_1,p_2,\ldots p_t}$ be a convex polygon in $\setkout(S)$, where $S$ is the set of $n$ points.
	Suppose that $\regnum(P,p)$ for each $p \in \outside(P)$ is given and a triangular range query can be computed in $\order(\log{n})$ time.
	Then, one can check whether $p_i \in V(P)$ is active in $\order(n\log{n})$ time.
\end{lemma}
\begin{proof}
It is sufficient to check whether or not $p_i = \lowerleft(I(\rmv(P,p_i)))$ from Lemma~\ref{lem:rmv-act}.
If we have the value $\regnum(\rmv(P,p_i),p)$ for each $p \in \outside(P)$,
we can list all the points in $I(\rmv(P,p_i))$. 
As a result, we can check whether or not $p_i = \lowerleft(I(\rmv(P,p_i)))$.
Hence, we consider how to compute $\regnum(\rmv(P,p_i),p)$ for each $p \in \outside(P)$ in this proof.
	
We first compute $\rmv(P,p_i)$ in $\order(n \log{n})$ time using a standard algorithm for finding convex hulls.
Next, we compute the tangency vertices $\tang_1(P,p)$ and $\tang_2(P,p)$ for each $p \in \outside(P)$.
Moreover, we compute $\tang_1(\rmv(P,p_i),p)$ and $\tang_2(\rmv(P,p_i),p)$ for each $p \in \outside(\rmv(P,p_i))$.
These can be done in $\order(n\log{n})$ time in total from Lemma~\ref{lem:tangents}.
By using Equation~(\ref{eq:points}), we compute $\regnum(\rmv(P,p_i),p)$ from $\regnum(P,p)$ for each $p \in \outside(\rmv(P,p_i))$, as follows.
\begin{listing}{aaa}
\item[Case~1:] $p = p_i$.
		
Then, we only set $\regnum(\rmv(P,p_i),p) = 0$, which takes constant time.
\item[Case~2:] $p\neq p_i$ and $p_i = \tang_1(P,p)$.
		
In this case, 
$p$ has a new tangency vertex $\tang_1'$ in $\rmv(P,p_i)$.
Note that $\tang_2(P,p)$ is still a tangency vertex of $p$ in $\rmv(P,p_i)$.
Then, we compute $\triangle_{\#}(p,\tang_1(P,p),\tang_1')$, which can be done in $\order(\log{n})$ time by a triangular range query,
and we set $\regnum(\rmv(P,p_i),p)=\regnum(P,p) - \triangle_{\#}(p,\tang_1(P,p),\tang_1')$.
\item[Case~3:] $p\neq p_i$ and $p_i = \tang_2(P,p)$.
		
The same discussion of Case~2 can be applied.
\item[Cases~4 and 5:] $p\neq p_i$, $p_i \neq \tang_1(P,p)$, and $p_i \neq \tang_2(P,p)$.
		
We check whether or not $p_i \in \triangle(p,\tang_1(P,p),\tang_2(P,p))$. 
If it is true, we set $\regnum(\rmv(P,p_i),p) = \regnum(P,p)+1$. 
Otherwise, we set $\regnum(\rmv(P,p_i),p) = \regnum(P,p)$.
This can be done in constant time.
\end{listing}
From the above case analysis, 
we can compute $\regnum(\rmv(P,p_i),p)$ for every $p\in \outside(P)$ in $\order(n\log{n})$ time in total.
Therefore, one can list the vertices in $\inspset(\rmv(P,p_i))$ and check whether $p_i = \lowerleft(\inspset(\rmv(P,p_i)))$ in $\order(n\log{n})$ time.
\end{proof}

The following lemma is immediate from the above lemma.
\begin{corollary}\label{cor:active_set_Type2}
	Let $P$ be a convex polygon in $\setkout(S)$, where $S$ is the set of $n$ points.
	Suppose that $\regnum(P,p)$ for each $p \in \outside(P)$ is given and a triangular range query can be computed in $\order(\log{n})$ time.
	Then, one can construct the set of active vertices in $P$ in $\order(n^2 \log{n})$ time.
\end{corollary}

Now, the following lemma estimates the running time for constructing a child.

\begin{lemma}\label{lem:child_const_Type2}
	Let $P$ be a polygon in $\setkout(S)$.
	Suppose that $\regnum(P,p)$ for each $p \in \outside(P)$ is given and a triangular range query can be computed in $\order(\log{n})$ time.
	Then, 
	given an active vertex $p_i \in V(P)$, one can construct $\rmv(P,p_i)$ and $\regnum(\rmv(P,p_i),p)$ for each $p \in \outside(\rmv(P,p_i))$ in $\order(n \log{n})$ time.
\end{lemma}
\begin{proof}
	By using a standard algorithm for finding convex hulls, we construct $\rmv(P,p_i)$ in $\order(n \log{n})$ time.
	To update $\regnum(\rmv(P,p_i),p)$ for each $p \in \outside(\rmv(P,p_i))$, we use Equation~(\ref{eq:points}) in the same way explained in the proof of Lemma~\ref{lem:update}, which can be done in $\order(n\log{n})$ time in total.
\end{proof}

%
%
\subsection{Enumeration algorithm}

Now, we are ready to describe our enumeration algorithm.
Algorithm~\ref{alg:main} is the main routine and Algorithm~\ref{alg:childenum} is a subroutine.

%
%
\begin{algorithm}[t]
	\tcc{Let $S$ be a set of $n$ points in Euclidean plane.}
	Construct the convex hull $\convexhull(S)$ of the input point set $S$\;
	Preprocess 
	$S$ for triangular range queries\;
	$\mathtt{FindChildren}(\convexhull(S),\emptyset)$\;
	\caption{$\mathtt{Enum}(S,k)$}
	\label{alg:main}
\end{algorithm}
%
%
\begin{algorithm}[t]
	\tcc{$P=\mseq{p_1,p_2,\ldots,p_t}$ is a polygon in $\setkout(S)$ with $t$~($k\le t$), and $S$ is stored in a global variable. If $P$ has at least one embeddable vertex $p_j=\largest(P)$. Otherwise $p_j=\emptyset$.}
	Output $P$\;
	Preprocess 
	$P$ for ray shooting queries\;
	\lIf{$p_j=\emptyset$}{$q=p_0$}
	\lElse{$q=\prede(p_j)$}
	\tcc{Note that $p_j=\largest(P)$ and $p_0$ is a sentinel vertex satisfying $p_0 \prec p_i$ for each $i=1,2,\ldots,t$.}
	\ForEach(\tcc*[h]{Enumeration of Type-1 children.}){point $p_i$ with $q \prec p_i$\label{line:type1_start}}{
		\ForEach{point $p \in \inside(P)$}{
			\lIf{the pair $(p_i,p)$ is active}{$\mathtt{FindChildren}(\dig(P,p_i,p),p))$}	
		}
	}\label{line:type1_end}
	\If{$P$ is convex and $\msize{\outside(P)}<k$}{\label{line:type2_start}
		\ForEach(\tcc*[h]{Enumeration of Type-2 children.}){vertex $p_i$ of $P$}{
			\If{$p_i$ is active}{
				Compute $\regnum(\rmv(P,p_i),p)$ for each $p \in \outside(\rmv(P,p_i))$ \label{line:regnum}\;
				$\mathtt{FindChildren}(\rmv(P,p_i),\emptyset))$\;
			}
		}
	}\label{line:type2_end}
	\caption{$\mathtt{FindChildren}(P=\mseq{p_1,p_2,\ldots ,p_t},p_j)$ \label{IR}}
	\label{alg:childenum}
\end{algorithm}
%
%

Algorithm~\ref{alg:main} first constructs the convex hull $\convexhull(S)$ of the input point set $S$.
Then, it executes a preprocess to $S$ for efficiently answering triangular range queries. 
Note that the preprocess for triangular range queries is executed only once in our algorithm.
Algorithm~\ref{alg:childenum} first outputs the current polygon $P \in \setkout(S)$ and executes a preprocess to $P$ for efficiently answering ray shooting queries. 
The preprocess is done once for a recursive call.
Next, the algorithm enumerates all the Type-1 children by dig operations in Lines~\ref{line:type1_start}--\ref{line:type1_end} and all the Type-2 children of $P$ by remove operations in Lines~\ref{line:type2_start}--\ref{line:type2_end}.

We have our main theorem.
\begin{theorem}\label{thm:enum}
	Let $S$ be a set of $n$ points in the Euclidean plane, and let $k$ be an integer with $0\leq k\leq n-3$. One can enumerate all the at most $k$-out polygons in $\setkout(S)$ in $\order(n^{2}\log n|\setkout(S)|)$ time and $\order(n^2)$ space.
\end{theorem}
\begin{proof}
	In each recursive call of $\mathtt{FindChildren}(P,p_j)$, triangular range and ray shooting queries are available by preprocesses.
	Also, from Line~\ref{line:regnum}, we can assume that $\regnum(P,p)$ for each $p\in \outside(P)$ is available
	if $P$ is convex.
	
	First, we assume that $P$ has $c_1$ Type-1 children.
	From Corollary~\ref{cor:active_set_Type1} and Lemma~\ref{lem:child_const_Type1}, Lines \ref{line:type1_start}--\ref{line:type1_end} takes $\order(n^2 \log n + c_1)$ time.
	Since $c_1$ is bounded above by $\order(n^2)$,
	each recursive call takes $\order(n^2 \log n)$ time for Lines \ref{line:type1_start}--\ref{line:type1_end}.
	
	Next, we assume that $P$ has $c_2$ Type-2 children.
	From Corollary~\ref{cor:active_set_Type2} and Lemma~\ref{lem:child_const_Type2}, Lines \ref{line:type2_start}--\ref{line:type2_end} takes $\order(n^2\log{n} + c_2(n \log{n}))$ time.
	Since $c_2$ is bounded above by $\order(n)$,
	each recursive call takes $\order(n^2 \log n)$ time for Lines \ref{line:type2_start}--\ref{line:type2_end}.
	
	Finally, we estimate the space complexity. 
    By Lemma~\ref{lem:triangle-query}, the data structure constructed on the set $S$ to support triangular range queries requires $\order(n^2)$ space. 
    This data structure is used in all recursive calls.
    The algorithm uses $\order(n)$ space in each recursive call for efficiently answering ray shooting queries.
	The depth of a family tree of $\setkout(S)$ is bounded above by $\order(n)$. Therefore, it is sufficient to use $\order(n^2)$ space for a stack of recursive calls.
\end{proof}

The theorem implies that 
the running time of our algorithm is output-polynomial.
Using the 
alternative output method
by Nakano and Uno~\cite{NakanoU05},
we have a polynomial-delay enumeration algorithm.
In the method,
the algorithm outputs polygons
with even depth when we go down a family tree
and outputs polygons with odd depth when we go up.
More precisely, we modify our algorithm
so that the algorithm outputs the current polygon $P$ before the children of $P$ in even depth and outputs $P$ after the children of $P$ in odd depth.
It is easy to see that the modified algorithm outputs a polygon once at most three edge traversals in a family.
See \cite{NakanoU05} for further details. 

\begin{corollary}
	Let $S$ be a set of $n$ points in the Euclidean plane, and let $k$ be an integer with $0\leq k\leq n-3$. One can enumerate all the at most $k$-out polygons in $\setkout(S)$ in $\order(n^{2}\log n)$ delay and and $\order(n^2)$ space.
\end{corollary}

\section{Conclusion}
In this paper, we studied the problem of listing all the at most $k$-out polygons of a given planar point set $S$, where integer $k>0$ is a part of the input. We proposed an algorithm for the problem which uses overall $\order(n^2)$ space and takes $\order(n^{2}\log n)$ time for enumerating each polygon in the output set $\setkout(S)$.

Future work includes proposing a more efficient algorithm.
Can we apply the reverse search with child lists~\cite{TeruiYHHKU23} to improve the running time?

\section*{Acknowledgments}

This work is partially supported by JSPS KAKENHI
Grant Numbers JP23K11027 and JP25K03080, 


\bibliographystyle{elsarticle-num} 
\bibliography{mybib}

\providecommand*\hyphen{-}
\begin{thebibliography}{10}
\expandafter\ifx\csname url\endcsname\relax
  \def\url#1{\texttt{#1}}\fi
\expandafter\ifx\csname urlprefix\endcsname\relax\def\urlprefix{URL }\fi
\expandafter\ifx\csname href\endcsname\relax
  \def\href#1#2{#2} \def\path#1{#1}\fi

\bibitem{Wasa16}
K.~Wasa, Enumeration of enumeration algorithms, arXiv abs/1605.05102 (2016).

\bibitem{AvisF96}
D.~Avis, K.~Fukuda, Reverse search for enumeration, Discrete Applied Mathematics 65~(1-3) (1996) 21--46.

\bibitem{B02}
S.~Bespamyatnikh, An efficient algorithm for enumeration of triangulations, Computational Geometry Theory and Applications 23~(3) (2002) 271--279.

\bibitem{KT09}
N.~Katoh, S.~Tanigawa, Enumerating edge-constrained triangulations and edge-constrained non-crossing geometric spanning trees, Discrete Applied Mathematics 157~(17) (2009) 3569--3585.

\bibitem{Wettstein17}
M.~Wettstein, Counting and enumerating crossing-free geometric graphs, Journal of Computational Geometry 8~(1) (2017) 47--77.

\bibitem{NakahataHMY20}
Y.~Nakahata, T.~Horiyama, S.~Minato, K.~Yamanaka, Compiling crossing-free geometric graphs with connectivity constraint for fast enumeration, random sampling, and optimization, arXiv abs/2001.08899 (2020).

\bibitem{YamanakaNMUN09}
K.~Yamanaka, S.~Nakano, Y.~Matsui, R.~Uehara, K.~Nakada, Efficient enumeration of pseudoline arrangements, in: Proceedings of European Workshop on Computational Geometry 2009, 2009, pp. 143--146.

\bibitem{HoriyamaS11}
T.~Horiyama, W.~Shoji, Edge unfoldings of platonic solids never overlap, in: Proceedings of the 23rd Annual Canadian Conference on Computational Geometry, 2011, pp. 65--70.

\bibitem{HoriyamaMS18}
T.~Horiyama, M.~Miyasaka, R.~Sasaki, Isomorphism elimination by zero-suppressed binary decision diagrams, in: S.~Durocher, S.~Kamali (Eds.), Proceedings of the 30th Canadian Conference on Computational Geometry, {CCCG} 2018, August 8-10, 2018, University of Manitoba, Winnipeg, Manitoba, Canada, 2018, pp. 360--366.

\bibitem{YamanakaAHOUY21}
K.~Yamanaka, D.~Avis, T.~Horiyama, Y.~Okamoto, R.~Uehara, T.~Yamauchi, Algorithmic enumeration of surrounding polygons, Discrete Applied Mathematics 303 (2021) 305--313.

\bibitem{TeruiYHHKU23}
S.~Terui, K.~Yamanaka, T.~Hirayama, T.~Horiyama, K.~Kurita, T.~Uno, Enumerating empty and surrounding polygons, {IEICE} Trans. Fundam. Electron. Commun. Comput. Sci. 106~(9) (2023) 1082--1091.

\bibitem{MarxM16}
D.~Marx, T.~Miltzow, Peeling and nibbling the cactus: Subexponential-time algorithms for counting triangulations and related problems, in: 32nd International Symposium on Computational Geometry, SoCG 2016, June 14-18, 2016, Boston, MA, {USA}, 2016, pp. 52:1--52:16.

\bibitem{AuerH96}
T.~Auer, M.~Held, Heuristics for the generation of random polygons, in: Proceedings of the 8th Canadian Conference on Computational Geometry, 1996, pp. 38--43.

\bibitem{Sohler99}
C.~Sohler, Generating random star-shaped polygons, in: Proceedings of the 11th Canadian Conference on Computational Geometry, 1999, pp. 174--177.

\bibitem{TeramotoMUA06}
S.~Teramoto, M.~Motoki, R.~Uehara, T.~Asano, Heuristics for generating a simple polygonalization, IPSJ SIG Technical Report 2006-AL-106(6), Information Processing Society of Japan (May 2006).

\bibitem{ZhuSSM96}
C.~Zhu, G.~Sundaram, J.~Snoeyink, J.~S.~B. Mitchell, Generating random polygons with given vertices, Computational Geometry: Theory and Applications 6 (1996) 277--290.

\bibitem{Eppstein24}
D.~Eppstein, Non-crossing hamiltonian paths and cycles in output-polynomial time, Algorithmica 86~(9) (2024) 3027–3053.

\bibitem{Bae22}
S.~W. Bae, Faster counting empty convex polygons in a planar point set, Inf. Process. Lett. 175 (2022) 106221.

\bibitem{MitchellRSW95}
J.~S.~B. Mitchell, G.~Rote, G.~Sundaram, G.~J. Woeginger, Counting convex polygons in planar point sets, Inf. Process. Lett. 56~(1) (1995) 45--49.

\bibitem{RoteWZW91}
G.~Rote, G.~J. Woeginger, B.~Zhu, Z.~Wang, Counting k-subsets and convex k-gons in the plane, Inf. Process. Lett. 38~(3) (1991) 149--151.

\bibitem{RoteW92}
G.~Rote, G.~J. Woeginger, Counting convex k-gons in planar point sets, Information Processing Letters 41~(4) (1992) 191--194.

\bibitem{DobkinEO90}
D.~P. Dobkin, H.~Edelsbrunner, M.~H. Overmars, Searching for empty convex polygons, Algorithmica 5~(4) (1990) 561--571.

\bibitem{WaseemY24}
W.~Akram, K.~Yamanaka, Enumerating at most k-out polygons, in: Proceedings of European Workshop on Computational Geometry 2024, 2024, pp. 13:1--13:8.

\bibitem{Preparata12}
F.~P. Preparata, M.~I. Shamos, Computational geometry: an introduction, Springer Science \& Business Media, 2012.

\bibitem{Preparata79}
F.~P. Preparata, An optimal real-time algorithm for planar convex hulls, Communications of the ACM 22~(7) (1979) 402--405.

\bibitem{GoswamiDN2004}
P.~P. Goswami, S.~Das, S.~C. Nandy, Triangular range counting query in {2D} and its application in finding $k$ nearest neighbors of a line segment, Computational Geometry 29~(3) (2004) 163 -- 175.

\bibitem{ChazelleEGGHSS94}
B.~Chazelle, H.~Edelsbrunner, M.~Grigni, L.~J. Guibas, J.~Hershberger, M.~Sharir, J.~Snoeyink, Ray shooting in polygons using geodesic triangulations, Algorithmica 12~(1) (1994) 54--68.

\bibitem{ChazelleG89}
B.~Chazelle, L.~J. Guibas, Visibility and intersection problems in plane geometry, Discrete {\&} Computational Geometry 4 (1989) 551--581.

\bibitem{NakanoU05}
S.~Nakano, T.~Uno, Generating colored trees, Proceedings of the 31th Workshop on Graph-Theoretic Concepts in Computer Science, (WG 2005) LNCS 3787 (2005) 249--260.

\end{thebibliography}




\end{document}